\documentclass[conference]{IEEEtran}
\IEEEoverridecommandlockouts

\usepackage{cite}
\usepackage{amsmath,amssymb,amsfonts}
 \usepackage{algorithm}
 \usepackage{array}
 \usepackage{ctable}
 \usepackage{makecell}
\usepackage{graphicx}
\usepackage{textcomp}
\usepackage{xcolor}
\usepackage{subfig}  
\usepackage{caption}
 \usepackage{algpseudocode}

\ifCLASSINFOpdf
\else
\fi
\usepackage{graphicx}
\usepackage{setspace}
\usepackage{verbatim}
\usepackage{multicol}
\usepackage{amsmath}
\usepackage{mathrsfs}
\usepackage{amssymb, amsthm}
\usepackage{bm}
\usepackage{cite}
\usepackage{enumerate}
\usepackage{diagbox}
\usepackage{xcolor}
\usepackage{setspace}
\usepackage{multirow}

\newcommand{\cM}{\mathcal{M}}

\newcommand{\cN}{\mathcal{N}}

\renewcommand{\IEEEQED}{\IEEEQEDopen}
\newtheorem{myth}{Theorem}

\newtheorem{myle}[myth]{\bf Lemma}

\def\BibTeX{{\rm B\kern-.05em{\sc i\kern-.025em b}\kern-.08em
    T\kern-.1667em\lower.7ex\hbox{E}\kern-.125emX}}
\usepackage{graphicx}
\usepackage{enumitem}
\setlist{nosep, leftmargin=*}
\allowdisplaybreaks[1]          
\usepackage{booktabs}
\begin{document}
\addtolength{\topmargin}{+0.5cm}
\title{
Multi-Waveguide Pinching Antenna Placement Optimization for Rate Maximization
}

\author{
\IEEEauthorblockN{Yue Zhang\IEEEauthorrefmark{1}, Yaru Fu\IEEEauthorrefmark{2}, Pei Liu\IEEEauthorrefmark{3}, Yalin Liu\IEEEauthorrefmark{2}, Kevin Hung\IEEEauthorrefmark{2}}
    \IEEEauthorblockA{\IEEEauthorrefmark{1}Department of Electronic Engineering, Shantou University, Shantou 515063, China
    }
    \IEEEauthorblockA{\IEEEauthorrefmark{2}School of Science and Technology, Hong Kong Metropolitan University, Hong Kong 999077, China
    }
    \IEEEauthorblockA{\IEEEauthorrefmark{3}School of Information Engineering, Wuhan University of Technology, Wuhan 430070, China\\
    E-mail: yuezhang@stu.edu.cn,  \{yfu, ylliu, khung\}@hkmu.edu.hk, pei.liu@ieee.org
    }
}

\maketitle

\begin{abstract}
Pinching antenna systems (PASS) have emerged as a technology that enables the large-scale movement of antenna elements, offering significant potential for performance gains in next-generation wireless networks. 
This paper investigates the problem of maximizing the average per-user data rate by optimizing the antenna placement of a multi-waveguide PASS, subject to a stringent physical minimum spacing constraint.
To address this complex challenge, which involves a coupled fractional objective and a non-convex constraint, we employ the fractional programming (FP) framework to transform the non-convex rate maximization problem into a more tractable one, and devise a projected gradient ascent (PGA)-based algorithm to iteratively solve the transformed problem.
Simulation results demonstrate that our proposed scheme significantly outperforms various geometric placement baselines, achieving superior per-user data rates by actively mitigating multi-user interference.
\end{abstract}


\section{Introduction}
The rapid evolution towards Beyond 5G (B5G) and 6G networks necessitates innovative hardware solutions capable of delivering ultra-high data rates, massive connectivity, and exceptional energy efficiency. 
In this landscape, pinching antenna systems (PASS) have emerged as a promising solution \cite{DingFlexible,LiuPinching}.
In contrast to other flexible antenna systems, such as the fluid antenna system \cite{WongFluid} and movable antenna system \cite{ZhuModeling}, PASS enables larger-scale movement of antenna elements.
Specifically, in PASS, each waveguide transmits signals through multiple dielectric clamps, which can be flexibly deployed based on the real-time user positions, and can hence effectively mitigate path loss and establish strong line-of-sight (LoS) links. 

Given the significant performance gains offered by dynamic placement, optimizing the position of pinching antennas has recently become a key research focus.
In \cite{XuRate}, a two-stage optimization algorithm was proposed to minimize the path loss and maximize the received signal strength of a single-waveguide PASS.
In \cite{QinJoint}, a reinforcement learning method was designed to improve the rate and sensing performance of a pinching antenna-assisted integrated sensing and communication system.
Later, closed-form expressions of the optimal positions of pinching antennas were derived in \cite{XieA} and \cite{DingAna} with the assumption that there is only one active pinching antenna in the system.
The study was further extended to multi-waveguide scenarios in \cite{ZhouA,FuPower,ChenDynamic,WangModeling}, where probability learning method \cite{ChenDynamic}, meta-learning approach \cite{ZhouA}, iterative algorithm \cite{FuPower}, and element-wise searching algorithm \cite{WangModeling} were respectively applied to solve the rate maximization or power minimization problems.

Although many efforts have been devoted to antenna placement optimization of PASS, a constraint that two neighbouring pinching antennas at the same waveguide should be separated with a minimum distance has not been well addressed.
This separation is primarily required to prevent mutual coupling between the antennas. 
Simultaneously, the pinching antenna's coupling length is critical because it fundamentally determines the antenna’s radiated power \cite{WangModeling}.
Therefore, in a practical system, a minimum spacing, at least larger than the coupling length of each antenna, must be enforced between two neighbouring antennas.
The authors of \cite{ChenDynamic} and \cite{WangModeling} took this constraint into account.
However, they addressed this issue by discretizing the waveguide into a finite number of candidate positions for pinching antennas to enforce minimum spacing.
While this method ensures the feasibility of the placement scheme, it sacrifices the flexibility of finding the truly optimal continuous antenna locations.
How to cope with the highly non-convex feasible set for the antenna positions created by the minimum antenna spacing constraint without using discretization still remains largely unknown.

Building upon the demonstrated potential of multi-waveguide PASS for enhanced capacity and interference mitigation, in this paper, we consider the antenna placement problem in a multi-waveguide PASS, where each waveguide is equipped with multiple antennas to simultaneously serve multiple users.
We explicitly consider the minimum antenna spacing constraint and design an efficient algorithm to address this issue, focusing on optimizing continuous antenna locations to avoid the performance loss associated with discretization methods used in prior works.
Our main contributions are summarized as follows.
\begin{itemize}
    \item We establish a system model for a downlink multi-waveguide PASS, explicitly relate the per-user data rate to the active pinching antennas' positions, and rigorously formulate the problem as a highly non-convex optimization problem subject to waveguide length and the critical minimum spacing constraint.
    \item We utilize the fractional programming (FP) framework to transform the non-convex antenna placement problem into a more tractable one, and design a projected gradient ascent (PGA) method to search for the optimal solution.
    \item We devise a novel projection algorithm that guarantees feasibility and prove that it achieves the minimum-distance projection property.
\end{itemize}


\section{System Model and Problem Formulation}
\begin{figure}
\centering
\includegraphics[width=3in]{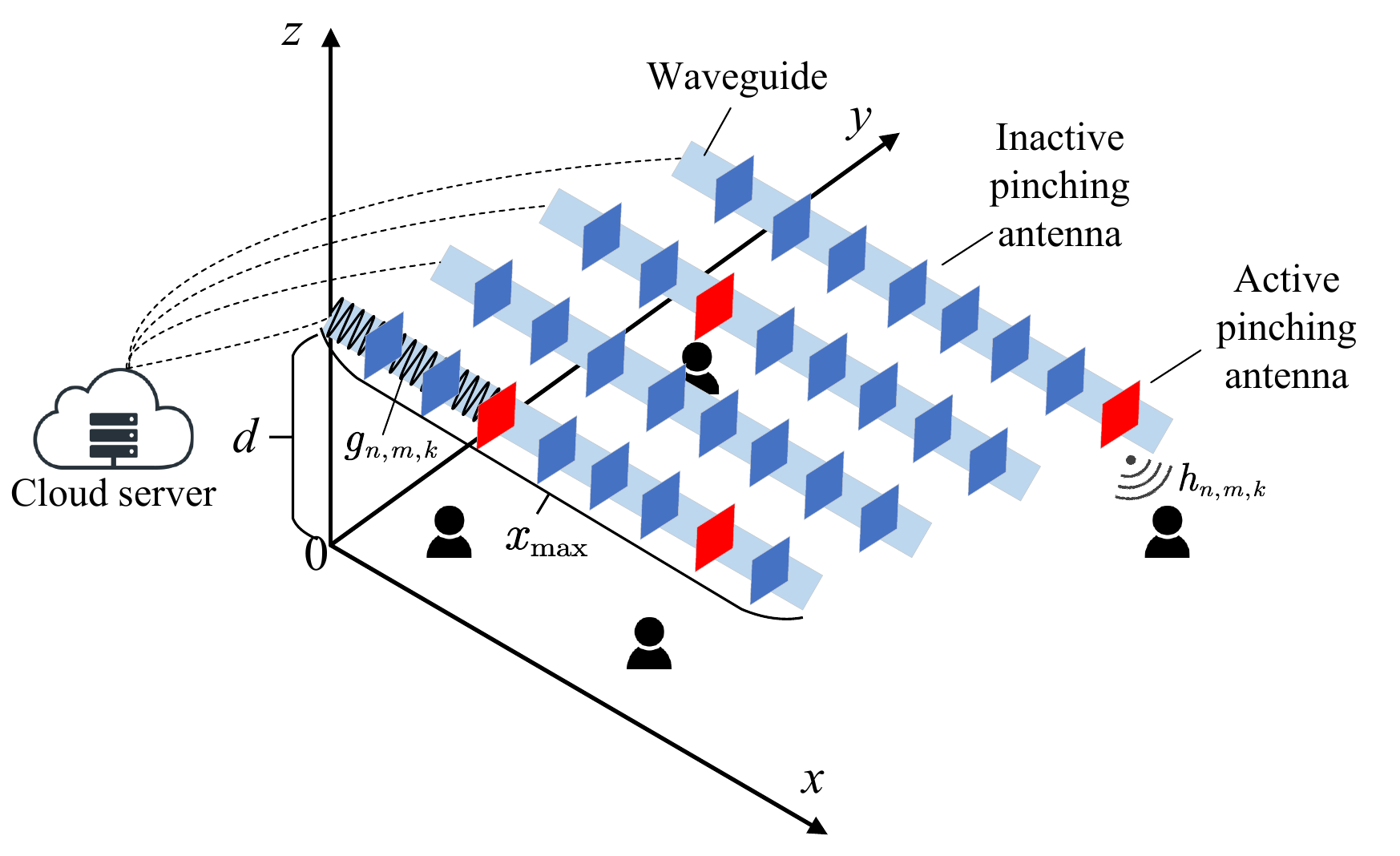}
\footnotesize
\captionsetup{font=footnotesize}
\caption{Graphic illustration of the pinching antenna system.}\label{PASS}
\end{figure}
We consider a downlink pinching antenna system consisting of a set of waveguides $\mathcal{N}$ serving for a set of single-antenna users $\mathcal{K}$, with $|\mathcal{N}|=N$ and $|\mathcal{K}|=K$.
Each waveguide is equipped with a set of pinching antennas, denoted by $\mathcal{M}_n$ for waveguide $n$ with $|\mathcal{M}_n| = M$ for all $n\in\mathcal{N}$.
We assume that all the waveguides are deployed in parallel to the $x$-axis in a Cartesian coordinate and at the same height $d$, as shown in Fig. \ref{PASS}, while all users are within the $xy$-plane with zero altitude.
Accordingly, we denote the position of user $k$ as $\bm{\psi}_k=(x_k,y_k,0)$ and the position of the $m$-th pinching antenna in $\mathcal{M}_n$ as $\bm{\psi}_{n,m}^{\mathrm{pin}}=(x_{n,m}^{\mathrm{pin}},y_n^{\mathrm{pin}},d)$.
All waveguides are equally spaced, and for waveguide $n$, $y_n^{\mathrm{pin}}=\frac{(n-1)y_{\max}}{N-1}$, where $y_{\max}$ is the $y$-coordinate of the farthest waveguide to the $x$-axis.

A cloud server is in charge of the baseband signal processing and feeds the superposition signal to each waveguide.
We assume that each user is only served by its closest waveguide, and let ${\cal{K}}_n$ denote the set of users served by waveguide $n$.
Therefore, the feeding signal of waveguide $n$ is $\sum_{k\in{\cal{K}}_n}s_{k}$, where $s_k$ denotes the signal transmitted to user $k$.
We assume a one-to-one mapping between the users served by waveguide $n$ and the active pinching antennas, such that the number of active antennas equals the number of users, i.e., $|\tilde{\mathcal{M}}_{n}|=|{\mathcal{K}}_{n}|$, where $\tilde{\mathcal{M}}_{n}\in \mathcal{M}_{n}$ denotes the set of the active pinching antennas of waveguide $n$. 

Compared to the path loss, the signal attenuation inside each waveguide is negligible.
Therefore, signals only undergo phase shifts, which are determined by the positions of pinching antennas, when traveling inside each waveguide.
Specifically, we denote the in-waveguide channel gain of the $m$-th pinching antenna at the $n$-th waveguide $g_{k,n,m}$ as
\begin{align}\label{WG_channel}
    g_{n,m,k} = e^{-j\frac{2\pi n_{\mathrm{eff}}}{\lambda}\|\bm{\psi}_{n,m}^{\mathrm{pin}}-\bm{\psi}_{n,0}^{\mathrm{pin}}\|},
\end{align}
where $\bm{\psi}_{n,0}^{\mathrm{pin}}=(0,y_n^{\mathrm{pin}},d)$ is the position of the feed point of waveguide $n$.
$\lambda$ is the wavelength of the transmit signal.
$n_{\mathrm{eff}}$ is the effective refractive index of a dielectric waveguide \cite{PozarMicro}.
The free-space channel gain between user $k$ and the $m$-th pinching antenna at waveguide $n$ is given by \cite{ZhangBeam}
\begin{align}\label{FS_channel}
    h_{n,m,k} = \frac{\lambda e^{-j\frac{2\pi}{\lambda}\|\bm{\psi}_k-\bm{\psi}_{n,m}^{\mathrm{pin}}\|}}{4\pi\|\bm{\psi}_k-\bm{\psi}_{n,m}^{\mathrm{pin}}\|}.
\end{align}
Thereby, the received signal of user $k$ can be expressed below:
\begin{align}
    y_k = &\underbrace{\sum_{m\in \tilde{\mathcal{M}}_{n_k}} h_{m,n_k,k}g_{m,n_k,k}P^{1/2}s_k}_{\mathrm{desired}\ \mathrm{signal}}\nonumber\\
    &+\underbrace{\sum_{i\in \mathcal{K}\backslash\{k\}}\sum_{m\in \tilde{\mathcal{M}}_{n_i}} h_{m,n_i,i}g_{m,n_i,i}P^{1/2}s_i}_{\mathrm{multiuser}\ \mathrm{interference}}+w_k,
\end{align}
where 
$P$ is the transmit power for each user.
$n_k$ is the index of the waveguide that serves user $k$.
$w_k$ is the additive white Gaussian noise received at user $k$ with zero mean and variance $\sigma^2$.
Correspondingly, the instantaneous achievable data rate of user $k$ can be obtained as\footnote{Without loss of generality, we normalize the bandwidth to one and focus on the spectral efficiency.}
\begin{align}\label{Rate}
    R_k = \log_2(1+\rho_k),
\end{align}
where
\begin{align}\label{SINR}
    \rho_k = \frac{P|\sum_{m\in \tilde{\mathcal{M}}_{n_k}} h_{n_k,m,k}g_{n_k,m,k}|^2}{P\sum_{i\in \mathcal{K}\backslash\{k\}}|\sum_{m\in \tilde{\mathcal{M}}_{n_i}} h_{n_i,m,k}g_{n_i,m,k}|^2+\sigma^2}
\end{align}
is the signal-to-interference-plus-noise ratio (SINR) of user $k$.
We also define the per-user data rate as
\begin{align}\label{ave_R}
    R=\frac{1}{K}\sum_{k=1}^K R_k.
\end{align}

From \eqref{FS_channel}--\eqref{ave_R}, we can see that the per-user data rate is determined by the positions of the activated pinching antennas.
In this paper, our goal is to maximize the per-user data rate by optimizing the positions of the activated pinching antennas based on the distribution of users.
The optimization problem can be formulated as
\begin{align}
    \text{P1: }&\mathop {\max} \limits_{\mathbf{x}^{\mathrm{pin}}}\,\,R\\
    \text{s.t.}\nonumber~
	&\text{C1}:~0\le x_{n,m}^{\mathrm{pin}}\le x_{\max},\forall n\in \cN, m\in \tilde{\cM}_n,\\
    &\text{C2}:~\| x_{n,m}^{\mathrm{pin}}- x_{n,m'}^{\mathrm{pin}}\|\ge d_{\min},\nonumber~\\ 
    &\ \ \ \ \ \ \ \forall n\in \cN, m,m'\in \tilde{\cM}_n\ \mathrm{and}\ m\ne m',\nonumber
\end{align}
where $x_{\max}$ is the length of one waveguide.
$\mathbf{x}^{\mathrm{pin}}$ denotes the collection of the $x$-coordinates of all the activated pinching antennas.
C2 implies that the distance between two neighboring pinching antennas should not be smaller than $d_{\min}$ to avoid the coupling effect.
We assume that a waveguide is long enough to satisfy that $x_{\max}\ge (K-1)d_{\min}$, and hence there always exists a feasible solution to P1.
 
The optimization problem P1 remains highly non-convex due to the complex, coupled nature of the objective function, as seen in \eqref{WG_channel}--\eqref{ave_R}, and the intricate non-convex feasible set $\cal X$ defined by constraints C1 and C2. 
This inherent non-convexity is a significant challenge, as it prevents the direct application of standard gradient-based optimization methods, which are only guaranteed to converge to the global optimum for convex problems, and would instead likely become trapped in sub-optimal local solutions \cite{WangModeling}. Therefore, to effectively tackle this difficult optimization and develop a practical, low-complexity algorithm for real-time deployment, we turn to the FP framework in the next section.
\vspace{-0.2cm}
\section{Algorithm Design}
\vspace{-0.1cm}
To address the highly non-convex structure of P1, we employ an FP approach \cite{ShenFrac}, which is a powerful technique for optimizing functions that are ratios of other functions, such as $\rho_k$ in \eqref{SINR}.
A major advantage of FP is its ability to transform a difficult ratio-based optimization problem into an equivalent but more tractable subtractive-form problem.
This transformation, as detailed in the following content, allows P1 to be equivalently transformed into a sequence of simpler problems that can be solved iteratively to achieve a locally optimal solution.

For ease of illustration, we let 
\begin{align}
A_k(\mathbf{x}^{\mathrm{pin}}) = {P\left|\textstyle\sum\nolimits_{m\in \tilde{\mathcal{M}}_{n_k}} h_{n_k,m,k}g_{n_k,m,k}\right|^2}    
\end{align}
and
\begin{align}
B_k(\mathbf{x}^{\mathrm{pin}})=    {P\textstyle\sum\nolimits_{i\in \mathcal{K}\backslash\{k\}}\left|\sum\nolimits_{m\in \tilde{\mathcal{M}}_{n_i}} h_{n_i,m,k}g_{n_i,m,k}\right|^2+\sigma^2}
\end{align}
where $\mathbf{x}^{\mathrm{pin}}$ is the collection of the $x$-coordinates of all the $K$ active pinching antennas.
Then we have
\begin{align}\label{Rk2}
R_k =\textstyle  \log_2\left(1+\frac{A_k(\mathbf{x}^{\mathrm{pin}})}{B_k(\mathbf{x}^{\mathrm{pin}})}\right).
\end{align}
From \eqref{Rk2}, we can clearly see that P1 is a sum-of-logarithms problem with a fractional parameter inside each logarithm function.
Given that $A_k(\mathbf{x}^{\mathrm{pin}})$ and $B_k(\mathbf{x}^{\mathrm{pin}})$ are both positive real function of $\mathbf{x}^{\mathrm{pin}}$, according to \cite[Theorem 3]{ShenFrac}, P1 can be equivalently reformulated as
\begin{align}
    \text{P2: }&\mathop {\max} \limits_{\mathbf{x}^{\mathrm{pin}},\bm{\gamma}}\,\,f(\mathbf{x}^{\mathrm{pin}},\bm{\gamma})\\
    \text{s.t.}\nonumber~
	&\text{C1\ and\ C2,}
\end{align}
where $\bm{\gamma}=(\gamma_1,\cdots,\gamma_K)$ is the collection of $K$ auxiliary variables, and the objective function $f(\mathbf{x}^{\mathrm{pin}},\bm{\gamma})$ is given by
\begin{align}\label{P4obj}
f(\mathbf{x}^{\mathrm{pin}},\bm{\gamma})=&\textstyle\sum_{k=1}^K \big{(}\ln{(1+\gamma_k)}-\gamma_k+\frac{(1+\gamma_k)A_k(\mathbf{x}^{\mathrm{pin}})}{A_k(\mathbf{x}^{\mathrm{pin}})+B_k(\mathbf{x}^{\mathrm{pin}})}\big{)}.
\end{align}
P2 can be solved by an iterative manner.
In each iteration, either $\mathbf{x}^{\mathrm{pin}}$ or $\bm{\gamma}$ is fixed and the other variable is optimized.
When $\mathbf{x}^{\mathrm{pin}}$ is fixed, $f(\mathbf{x}^{\mathrm{pin}},\bm{\gamma})$ is a concave function of $\bm{\gamma}$.
Therefore, by using the first-order optimality condition, i.e., $\frac{\partial f(\mathbf{x}^{\mathrm{pin}},\bm{\gamma})}{\partial \gamma_k}=0$, the optimal $\gamma_k$ can be obtained as
\begin{align}
\gamma_k^* =\textstyle \frac{A_k(\mathbf{x}^{\mathrm{pin}})}{B_k(\mathbf{x}^{\mathrm{pin}})},\ k=1,\cdots, K.
\end{align}
However, for a fixed $\bm{\gamma}$, solving P2 is still not an easy task due to the fractional form of $\frac{A_k(\mathbf{x}^{\mathrm{pin}})}{B_k(\mathbf{x}^{\mathrm{pin}})}$.
To further simplify the problem, we apply the Corollary 1 in \cite{ShenFrac1} to equivalently transform P2 to
\begin{align}
    \text{P3: }&\mathop {\max} \limits_{\mathbf{x}^{\mathrm{pin}},\bm{\gamma},\bm{\zeta}}\,\,F(\mathbf{x}^{\mathrm{pin}},\bm{\gamma},\bm{\zeta})\\
    \text{s.t.}\nonumber~
	&\text{C1\ and\ C2,}
\end{align}
where $\bm{\zeta}=(\zeta_1,\cdots,\zeta_K)$ is another collection of $K$ auxiliary variables, and the objective function of P3 is given by
\begin{align}\label{g_function}
&F(\mathbf{x}^{\mathrm{pin}},\bm{\gamma},\bm{\zeta})=\textstyle \sum_{k=1}^K\Big{(} \ln{(1+\gamma_k)}-\gamma_k \nonumber\\
    &{+}\textstyle 2\zeta_k\sqrt{(1+\gamma_k)A_k(\mathbf{x}^{\mathrm{pin}})}{-}\zeta_k^2\big{(}A_k(\mathbf{x}^{\mathrm{pin}}){+}B_k(\mathbf{x}^{\mathrm{pin}})\big{)}\Big{)}.  
\end{align}
Given that $F(\mathbf{x}^{\mathrm{pin}},\bm{\gamma},\bm{\zeta})$ is also a concave function of $\zeta_k,\ k=1,\cdots,K$, if $\mathbf{x}^{\mathrm{pin}}$ and $\bm{\gamma}$ are fixed, the optimal $\zeta_k,\ k=1,\cdots,K$, can be derived as
\begin{align}
\zeta_k^* =\textstyle \frac{\sqrt{(1+\gamma_k)A_k(\mathbf{x}^{\mathrm{pin}})}}{A_k(\mathbf{x}^{\mathrm{pin}})+B_k(\mathbf{x}^{\mathrm{pin}})},\ k=1,\cdots,K.    
\end{align}

The problem now becomes how to obtain the optimal $\mathbf{x}^{\mathrm{pin}}$ with given $\bm{\gamma}$ and $\bm{\zeta}$.
As P3 is a non-convex problem with respect to $\mathbf{x}^{\mathrm{pin}}$, it is difficult to derive its global optimal solution.
Considering the real-time deployment requirement of pinching antennas, we propose to search the optimal $\mathbf{x}^{\mathrm{pin}}$ via a low-complexity projected gradient ascent (PGA) method.
The main step of the PGA method is
\begin{align}\label{PGA}
 \mathbf{x}^{\mathrm{pin}}(\tau+1)  = \Phi_{\mathcal{X}}\big{(}\mathbf{x}^{\mathrm{pin}}(\tau)+\bm{\omega}_{\mathbf{x}^{\mathrm{pin}}}(\tau)\cdot\mu(\tau,t)\big{)}, 
\end{align}
where $\Phi_{\mathcal{X}}(\mathbf{x}^{\mathrm{pin}})$ is the projection function that projects $\mathbf{x}^{\mathrm{pin}}$ onto its feasible set $\mathcal{X}=\{\mathbf{x}^{\mathrm{pin}}\in[0,x_{\max}]^K\big{|}|x_{n,j}^{\mathrm{pin}}-x_{n,k}^{\mathrm{pin}}|\ge d_{\min},\forall j, k\in \tilde{\mathcal{M}}_n,n\in{\mathcal{N}}\}$.
$\mu(\tau,t)$ is the step size at iteration $\tau$ of the inner loop and iteration $t$ of the outer loop.
$\bm{\omega}_{\mathbf{x}^{\mathrm{pin}}}(\tau)=\frac{\partial F(\mathbf{x}^{\mathrm{pin}}(\tau),\bm{\gamma},\bm{\zeta})}{\partial \mathbf{x}^{\mathrm{pin}}(\tau)}$ is the gradient of $F(\mathbf{x}^{\mathrm{pin}}(\tau),\bm{\gamma},\bm{\zeta})$ with respect to $\mathbf{x}^{\mathrm{pin}}(\tau)$, with its detailed derivation presented in Appendix A.
As the feasible set $\mathcal{X}$ is non-convex because of the minimum distance constraint C2, we devise the projection algorithm in Algorithm \ref{Projection} to ensure that the updated $\mathbf{x}^{\mathrm{pin}}$ is feasible and closest to the non-projected solution.

Since the $\mathbf{x}^{\mathrm{pin}}$ coordinates for different waveguides $n\in {\cal N}$ are independent of each other, the projection $\Phi_{\mathcal{X}}$ can be applied separately to the set of antennas within each waveguide.
Crucially, for a given waveguide, the projection operation is equivalent to finding the closest ordered sequence of antenna locations that maintains a minimum separation $d_{\min}$.
This specific structure allows us to design an efficient pairwise-based projection algorithm that moves the antennas the minimum required distance to satisfy the constraints.

Algorithm \ref{Projection} is designed based on the following lemma and theorem.
\begin{myle}
Given the sorting function $\beta:\{1,\cdots,|\tilde{\mathcal{M}}_n| \}\rightarrow \{1,\cdots,|\tilde{\mathcal{M}}_n| \}$ that maps the original indices of $\{x^{\mathrm{pin}}_{m,n}\}_{m\in \tilde{\mathcal{M}}_n}$ to their new indices in ascending order, i.e., $x^{\mathrm{pin}}_{\beta(1),n}\le x^{\mathrm{pin}}_{\beta(2),n}\le\cdots\le x^{\mathrm{pin}}_{\beta(|\tilde{\mathcal{M}}_n|),n}$, the minimum distance constraint C2 is equivalent to requiring that the distance between any two antennas $m$ and $j$ must satisfy:
\begin{align}
    |x^{\mathrm{pin}}_{\beta(j),n}-x^{\mathrm{pin}}_{\beta(m),n}|\ge |j-m|d_{\min}.
\end{align}
This means that to satisfy all pairwise separation constraints, the total span occupied by $k$ consecutive antennas must be at least $(k-1)d_{\min}$.
\end{myle}
\noindent The proof is straightforward and hence omitted here.
Building on Lemma 1, the following theorem demonstrates the effectiveness of Algorithm \ref{Projection}.
\begin{myth}
The output of Algorithm \ref{Projection} is a feasible solution to P1 and satisfies
\begin{align}\label{OptProj}
\mathbf{x}^{\mathrm{proj}} = \underset{\mathbf{x}}{\arg\min}\{\|\mathbf{x}-\mathbf{x}^{\mathrm{pin}}\|\big{|}\mathbf{x}\in{\mathcal{X}}\}.    
\end{align}
\end{myth}
\begin{proof}
When constraints C1 and C2 are satisfied, we have
\begin{align}\label{Cond1}
0\le x^{\mathrm{pin}}_{\beta(1),n}\le x^{\mathrm{pin}}_{\beta(2),n}\le \cdots\le 
x^{\mathrm{pin}}_{\beta(|\tilde{\mathcal{M}}_n|),n}\le x_{\max},
\end{align}
\begin{align}\label{Cond2}
x^{\mathrm{pin}}_{\beta(m),n}+d_{\min}\le x^{\mathrm{pin}}_{\beta(m+1),n},\forall m=1,\cdots,|\tilde{\mathcal{M}}_n|-1,   
\end{align}
\begin{align}\label{Cond3}
\textstyle x^{\mathrm{pin}}_{\beta(|\tilde{\mathcal{M}}_n|+1-m),n}-x^{\mathrm{pin}}_{\beta(m),n}\ge (|\tilde{\mathcal{M}}_n|+1-2m)d_{\min},\nonumber\\
\textstyle\forall m=1,\cdots,\Big{\lceil}\frac{|\tilde{\mathcal{M}}_n|}{2}\Big{\rceil}.    
\end{align}
\eqref{Cond1} and \eqref{Cond2} are the necessary and sufficient conditions of constraints C1 and C2, respectively, while \eqref{Cond3} is the necessary condition of constraint C2.

\begin{algorithm}[t!]
\small
\caption{Projection of $\mathbf{x}^{\mathrm{pin}}$ onto $\mathcal X$}\label{Projection}
\begin{algorithmic}[1]
\Require $\mathbf{x}^{\mathrm{pin}}$.
\For {$n\in {\mathcal{N}}$}
\State Apply the sorting function $\beta:\{1,\cdots,|\tilde{\mathcal{M}}_n| \}\rightarrow \{1,\cdots,|\tilde{\mathcal{M}}_n| \}$ to the indices of the elements of $\mathbf{x}^{\mathrm{pin}}$.
\State $x_{\mathrm{low}} = 0$, $x_{\mathrm{up}} = x_{\max}$.
\For {$m=1,\cdots,\big{\lceil}\frac{|\tilde{\mathcal{M}}_n|}{2}\big{\rceil}$}
\State $j = |\tilde{\mathcal{M}}_n|+1-m$.
\State $x^{\mathrm{proj}}_{\beta(m),n}\leftarrow\max(x_{\mathrm{low}},x^{\mathrm{pin}}_{\beta(m),n})$.
\State $x^{\mathrm{proj}}_{\beta(j),n}\leftarrow\min(x_{\mathrm{up}},x^{\mathrm{pin}}_{\beta(j),n})$.
\If{$x^{\mathrm{proj}}_{\beta(j),n}-x^{\mathrm{proj}}_{\beta(m),n}<(j-m)d_{\min}$}
\State $\Delta x = (j-m)d_{\min}-(x^{\mathrm{proj}}_{\beta(j),n}-x^{\mathrm{proj}}_{\beta(m),n})$.
\If{$x^{\mathrm{proj}}_{\beta(m),n}-\Delta x< x_{\mathrm{low}}$}
\State $x^{\mathrm{proj}}_{\beta(j),n}{\leftarrow}x^{\mathrm{proj}}_{\beta(j),n}{+}(\Delta x{-}(x^{\mathrm{proj}}_{\beta(m),n}{-}x_{\mathrm{low}}))$.
\State $x^{\mathrm{proj}}_{\beta(m),n}\leftarrow x_{\mathrm{low}}$.
\Else
\State $x^{\mathrm{proj}}_{\beta(m),n}\leftarrow x^{\mathrm{proj}}_{\beta(m),n}-\Delta x$.
\EndIf
\EndIf
\State $x_{\mathrm{low}}\leftarrow x^{\mathrm{proj}}_{\beta(m),n}+d_{\min}$.
\State $x_{\mathrm{up}}\leftarrow x^{\mathrm{proj}}_{\beta(j),n}-d_{\min}$.
\EndFor
\State Apply the inverse permutation of the initial sorting $\beta^{-1}:\{1,\cdots,|\tilde{\mathcal{M}}_n| \}\rightarrow \{1,\cdots,|\tilde{\mathcal{M}}_n| \}$ to the indices of $\{x^{\mathrm{proj}}_{\beta(m),n}\}_{m\in \tilde{\mathcal{M}}_n}$.
\EndFor
\Ensure The feasible solution $\mathbf{x}^{\mathrm{proj}}$.
\end{algorithmic}
\end{algorithm}

Let us first prove the feasibility of the output of Algorithm 1.
The algorithm's iterative approach processes the antennas in pairs, starting from the outermost pair $(x^{\mathrm{pin}}_{\beta(1),n},x^{\mathrm{pin}}_{\beta(|\tilde{\mathcal{M}}_n|),n})$ and move inwards.
At each stage $m$, the antennas are first projected onto the boundary limits $x_{\mathrm{low}}$ and $x_{\mathrm{up}}$.
By setting their initial values as $0$ and $x_{\max}$, respectively, as shown in Step 3, the constraint C1 is satisfied.
Steps 8 and 9 show that if the total span of the current pair violates \eqref{Cond3}, the algorithm adjusts the positions by $\Delta x$.
This adjustment ensures that the minimum distance is satisfied, first by shifting the lower antenna (Step 14) or, if this violates the lower bound, by shifting both the lower and upper antennas to respect the boundary (Steps 11 and 12).
This process is executed for all pairs, and the update of $x_{\mathrm{low}}$ and $x_{\mathrm{up}}$ (Steps 17 and 18) ensures that subsequent inner pair satisfies \eqref{Cond2} with respect to the already fixed outer pair.
Consequently, all three conditions of the feasible set, \eqref{Cond1}-\eqref{Cond3}, are satisfied by the final output, proving the feasibility.

Specifically, whenever a constraint is violated, the antennas are moved only the distance $\Delta x$ necessary to satisfy the lower bound $(j-m)d_{\min}$.
This minimal movement strategy is equivalent to finding the $\mathbf{x}^{\mathrm{proj}}$ that minimizes the Euclidean distance $\|\mathbf{x}^{\mathrm{proj}}-\mathbf{x}^{\mathrm{pin}}\|$ subject to the constraints, which proves \eqref{OptProj}.
\end{proof}

With Algorithm 1, we summarize our FP-based iterative method for solving P3 in Algorithm \ref{SolPro5}, which converges to a stationary point of the original problem P1.
\begin{algorithm}[t!]
\small
\caption{FP-Based Iterative Approach for P3}\label{SolPro5}
\begin{algorithmic}[1]
\Require A feasible solution of P3 $\mathbf{x}^{\mathrm{pin}}(0)$.
\State $t=0$
\Repeat
    \State $\gamma_k(t)=\frac{A_k(\mathbf{x}^{\mathrm{pin}}(t))}{B_k(\mathbf{x}^{\mathrm{pin}}(t))},\ k=1,\cdots, K.$ 
    \State $\zeta_k(t)=\frac{\sqrt{(1+\gamma_k(t))A_k(\mathbf{x}^{\mathrm{pin}}(t))}}{A_k(\mathbf{x}^{\mathrm{pin}}(t))+B_k(\mathbf{x}^{\mathrm{pin}}(t))},\ k=1,\cdots,K.$
    \State $\tau=0$, $\mathbf{x}^{\mathrm{pin}}(\tau)=\mathbf{x}^{\mathrm{pin}}(t)$.
    \Repeat
        \State Update $\mathbf{x}^{\mathrm{pin}}(\tau+1)$ using \eqref{PGA} and Algorithm \ref{Projection}.
        \State $\tau\leftarrow \tau+1$.
    \Until Convergence.
    \State $\mathbf{x}^{\mathrm{pin}}(t+1)=\mathbf{x}^{\mathrm{pin}}(\tau+1)$.
    \State $t\leftarrow t+1$.
\Until Convergence.
\Ensure The solution $\mathbf{x}^{\mathrm{pin}}(t+1)$, $\bm{\gamma}(t)$ and $\bm{\zeta}(t)$.
\end{algorithmic}
\end{algorithm}
\section{Simulation Results}
In this section, extensive simulations are conducted to verify the performance of the proposed algorithms.
Unless otherwise stated, the values of different parameters are summarized in Table I.
All users are randomly distributed in a rectangular region $[0,x_{\max}]\times[0,y_{\max}]$.
The average per-user data rate is obtained using the Monte-Carlo method and averaged over 1000 samples of users' positions.
For Algorithm 2, the expression of the step size $\mu(\tau,t)$ in \eqref{PGA} is set as 
\begin{align}
\mu(\tau,t) =\textstyle \frac{0.01}{(\tau+\tau_{\max}(t-1))^{0.6}},
\end{align}
where $\tau_{\max}$ is the maximum number of iterations in the inner loop of Algorithm 2.
All experiments were conducted on a system equipped with an Intel i7-13700 CPU, utilizing MATLAB R2025a for implementation.

For illustration, we consider the following three benchmarks for comparison:
\begin{itemize}
    \item \textbf{Closest-to-User Placement (CUP):} The active antenna coordinates $\mathbf{x}^{\mathrm{pin}}$ are determined by projecting the user $x$-coordinates onto the feasible set $\cal X$, ensuring the final placement is the closest possible feasible solution to the users' physical locations.
    \item \textbf{Uniform Pre-Placement with Closest Selection (UPCS):} Each waveguide is populated with a fixed grid of antennas with spacing interval $0.1$. The $x$-coordinates of the final active antenna set $\mathbf{x}^{\mathrm{pin}}$ are formed by selecting the single pre-placed antenna that is closest to each serving user.
    \item \textbf{Random Pre-Placement with Closest Selection (RPCS):} Each waveguide has $M=K$ antennas randomly placed according to a uniform distribution, while still satisfying constraint C2. The $x$-coordinates of the final active antenna set $\mathbf{x}^{\mathrm{pin}}$ are formed by selecting the single pre-placed antenna that is closest to each serving user.
\end{itemize}

\begin{table}
    \centering
    \caption{Simulation Setting}
    \label{SimSet}
  \scalebox{0.8}{  \begin{tabular}{|c|c|}
    \hline
    \textbf{Parameter} & \textbf{Value}\\
    \hline
    Carrier frequency     & $28$ GHz \\ 
     \hline
      \makecell[c]{Simulation Region Width\\ (waveguide length) $x_{\max}$}     &$10$ m \\ 
     \hline
     Simulation Region length $y_{\max}$     &$10$ m \\ 
     \hline
    Height of waveguides $h$     &$3$ m \\ 
     \hline
    Number of waveguides $N$     &$[5,10]$ \\ 
     \hline
    Number of users $K$     &$[50,100]$ \\ 
     \hline
    Effective refractive index $n_{\mathrm{eff}}$     & $1.4$ \\
     \hline
    Noise power  $\sigma^2$    & $-90$ dBm\\ 
     \hline
    Transmit power for each user     & $1$ Watt\\
     \hline
    \makecell[c]{The minimum distance between two \\ pinching antennas in one waveguide $d_{\min}$}     & $\{\lambda/2,0.1.0.2\}$\\
    \hline
    \makecell[c]{The maximum number of iterations in\\ the inner loop of Algorithm 1 $\tau_{\max}$}  & $100$ \\ 
     \hline
    \end{tabular}   }
\end{table}
\begin{figure}[t!]
\centering
\includegraphics[width=2.4in]{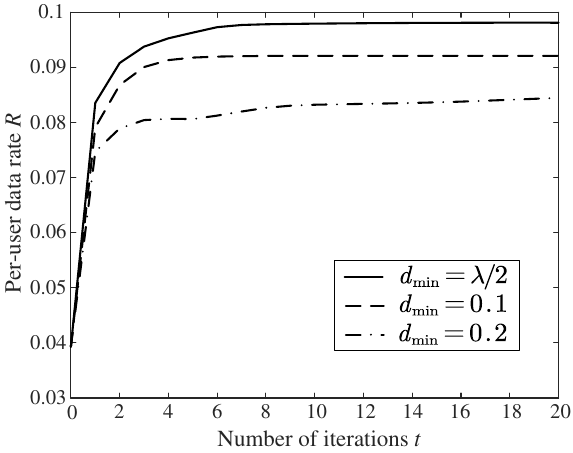}
\footnotesize
\captionsetup{font=footnotesize}
\caption{Per-user data rate versus the number of outer iterations $t$ of Algorithm 2, wherein $N=6$, $K=50$.}\label{conv_alg1}
\end{figure}

Fig. \ref{conv_alg1} shows the convergence performance of Algorithm 2 by plotting the per-user data rate $R$ against the number of outer iterations $t$. 
The algorithm exhibits fast convergence, reaching a stable solution within approximately $5$ to $10$ outer iterations across all tested minimum distance constraints $d_{\min}$.
Specifically, the algorithm converges quickest when the constraint is least restrictive, i.e., $d_{\min} = \lambda/2\approx0.01$, and achieves the highest data rate.
We also observe that as $d_{\min}$ increases from $\lambda/2$ to $0.1$ and $0.2$, the final achieved per-user data rate decreases, and the convergence speed becomes slower. 
This is because a larger $d_{\min}$ imposes a more restrictive minimum distance requirement between antennas, making the feasible set $\cal X$ for the antenna positions more constrained and highly non-convex. 
It is also worth noting that, although in the inner loop of Algorithm 1, we set the maximum number of iterations $\tau_{\max}=100$, the operation time in our simulation environment is about $30$ ms for these $100$ iterations.
Therefore, the proposed algorithm is computationally efficient enough for real-time applications.
Based on Fig. \ref{conv_alg1}, we set the maximum number of outer iterations $t_{\max}=10$ in Figs. \ref{AveR_vs_K} and \ref{AveR_vs_N}.

Fig. \ref{AveR_vs_K}  illustrates the relationship between the average per-user data rate and the number of users $K$.
We can see that the proposed scheme consistently achieves a significantly higher average per-user data rate than the three baseline schemes across all tested numbers of users, confirming the superiority of the FP approach over heuristic strategies.
Specifically, when $K=50$ and $d_{\min}=0.1$, the performance gain achieves $113\%$, $113\%$, and $127\%$ over CUP, UPCS, and RPCS schemes, respectively.
For all the schemes, the average per-user data rate decreases as the number of users $K$ increases. 
This is because, as $K$ grows, the total interference in the system increases, leading to a drop in the SINR for each user. Simultaneously, more active antennas must be placed in a limited waveguide length, which severely constricts the feasible set $\cal X$ due to constraint C2, preventing the placement of antennas in optimal, high-rate positions. 
It is also interesting to see that the CUP scheme, where pinching antennas are placed at the same $x$-coordinate as their serving users, maintains a relatively low average per-user data rate. 
This observation demonstrates that simply placing antennas close to users is not an effective strategy, as desired signals may add destructively at each user.

\begin{figure}[t!]
\centering
\includegraphics[width=2.4in]{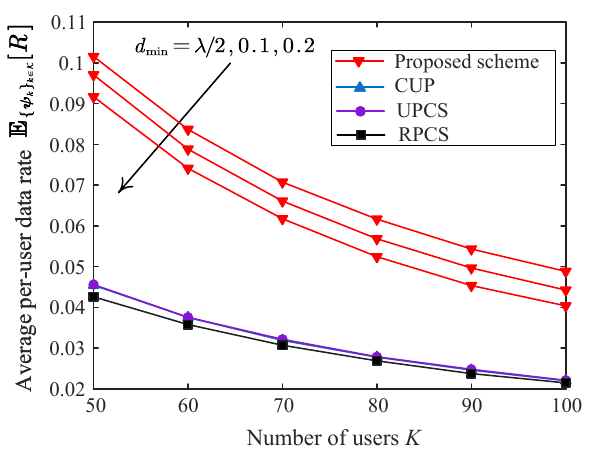}
\footnotesize
\captionsetup{font=footnotesize}
\caption{Average per-user data rate versus the number of users $K$, where  the number of waveguides is set $N=6$.}\label{AveR_vs_K}
\end{figure}

\begin{figure}[t!]
\centering
\includegraphics[width=2.4in]{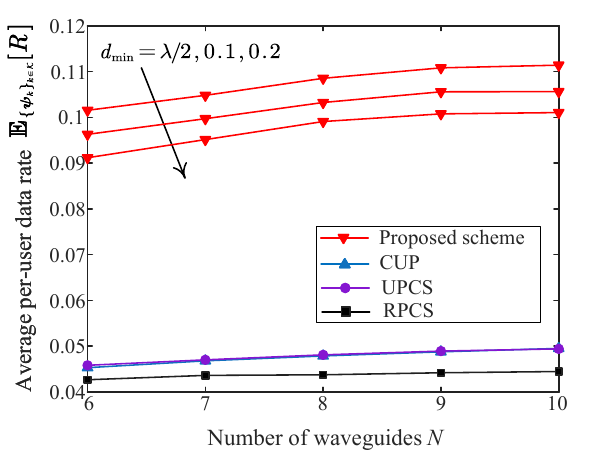}
\footnotesize
\captionsetup{font=footnotesize}
\caption{Average per-user data rate versus the number of waveguides $N$. $K=50$.}\label{AveR_vs_N}
\end{figure}

Fig. \ref{AveR_vs_N} demonstrates how the average per-user data rate varies with the number of waveguides $N$. 
The proposed scheme consistently and significantly outperforms the three baseline schemes across all simulated numbers of waveguides.
In contrast to the effect of $K$, the average per-user data rate increases as the number of waveguides $N$ increases for all schemes.
When $N$ increases and $K$ is constant, the number of users per waveguide decreases, which effectively localizes and reduces the total multi-user interference that each user experiences. This expanded spatial resource provides the proposed scheme with a greater opportunity to select optimal antenna positions that balance signal focusing and interference cancellation.
\section{Conclusion}
In this paper, we addressed the highly non-convex antenna placement optimization problem for a multi-waveguide pinching antenna system.
Our approach utilized the FP framework to transform the objective into a more tractable form, which was then solved with a specialized projection step to enforce feasibility.
Simulation results demonstrated that the proposed scheme significantly outperforms all geometric baselines.
We further showed that the data rate decreases with the number of users $K$ due to heightened interference and tighter constraints, but improves with the number of waveguides $N$ as interference is better mitigated and spatial resources are distributed.
For future work, it would be interesting to extend this framework to scenarios involving joint optimization of antenna placement and waveguide deployment.
\begin{appendices}
\section{Derivation of $\bm{\omega}_{\mathbf{x}^{\mathrm{pin}}}(\tau)$}
In the following, we omit the iteration index $\tau$ for the sake of conciseness.
Based on \eqref{g_function}, we have
\begin{align}\label{grd_start_x}
\bm{\omega}_{\mathbf{x}^{\mathrm{pin}}}= &\textstyle\sum_{k=1}^K \Big{(}\big{(}\zeta_k\sqrt{(1+\gamma_k)}A_k^{-\frac{1}{2}}(\mathbf{x}^{\mathrm{pin}})-\zeta_k^2\big{)}\nonumber\\
&\textstyle\cdot\frac{\partial A_k(\mathbf{x}^{\mathrm{pin}})}{\partial \mathbf{x}^{\mathrm{pin}}}-\zeta_k^2\frac{\partial B_k(\mathbf{x}^{\mathrm{pin}})}{\partial \mathbf{x}^{\mathrm{pin}}}\Big{)}   
\end{align}
Let 
\begin{align}\label{A_aux}
    A_{i,k} =\textstyle P\left|\sum\nolimits_{m\in \tilde{\mathcal{M}}_{n_i}} h_{n_i,m,k}g_{n_i,m,k}\right|^2.
\end{align}
Then we have
\begin{align}
\textstyle\frac{\partial A_k(\mathbf{x}^{\mathrm{pin}})}{\partial \mathbf{x}^{\mathrm{pin}}}= \frac{\partial A_{k,k}}{\partial \mathbf{x}^{\mathrm{pin}}},  
\end{align}
and
\begin{align}
\textstyle\frac{\partial B_k(\mathbf{x}^{\mathrm{pin}})}{\partial \mathbf{x}^{\mathrm{pin}}}= \sum_{i\in \mathcal{K}\backslash\{k\}}\frac{\partial A_{i,k}}{\partial \mathbf{x}^{\mathrm{pin}}}.  
\end{align}
Knowing that for a complex function $f(x)$ of $x$, the derivative 
\begin{align}\label{complex_grd}
\textstyle\frac{\mathrm{d}|f(x)|^2}{\mathrm{d} x}=2\Re\left(f^*(x)\cdot \frac{\mathrm{d}f(x)}{\mathrm{d} x}\right), 
\end{align}
the expression of $\frac{\partial A_{i,k}}{\partial \mathbf{x}^{\mathrm{pin}}}$ can be obtained as
\begin{align}\label{A_grd_x}
\textstyle\frac{\partial A_{i,k}}{\partial \mathbf{x}^{\mathrm{pin}}}=&\textstyle2P\Re{\Big(}{\big(}\sum\nolimits_{m\in \tilde{\mathcal{M}}_{n_i}} h^*_{n_i,m,k}g^*_{n_i,m,k} {\big)}\nonumber\\
&\textstyle\cdot \sum\nolimits_{m\in \tilde{\mathcal{M}}_{n_i}}g_{n_i,m,k}\frac{\partial h_{n_i,m,k}}{\partial \mathbf{x}^{\mathrm{pin}}}+ h_{n_i,m,k}\frac{\partial g_{n_i,m,k}}{\partial \mathbf{x}^{\mathrm{pin}}}{\Big)}.    
\end{align}
Based on \eqref{FS_channel}, $\frac{\partial h_{n_i,m,k}}{\partial \mathbf{x}^{\mathrm{pin}}}$ can be derived as
\begin{align}
    \textstyle\frac{\partial h_{n_i,m,k}}{\partial \mathbf{x}^{\mathrm{pin}}}=&\textstyle h_{n_i,m,k}\cdot{\Big(}{\big(}-j\frac{2\pi}{\lambda}-\frac{1}{\|\bm{\psi}_k-\bm{\psi}_{n_i,m}^{\mathrm{pin}}\|}{\big)}\frac{\partial \|\bm{\psi}_k-\bm{\psi}_{n_i,m}^{\mathrm{pin}}\|}{\partial \mathbf{x}^{\mathrm{pin}}}{\Big)}.
\end{align}
Given that 
\begin{align}\label{Dis_p_u}
\|\bm{\psi}_k-\bm{\psi}_{n_i,m}^{\mathrm{pin}}\|=\sqrt{(x_k-x_{n_i,m}^{\mathrm{pin}})^2+(y_k-y_{n_i}^{\mathrm{pin}})^2+d^2},    
\end{align}
$\frac{\partial \|\bm{\psi}_k-\bm{\psi}_{n_i,m}^{\mathrm{pin}}\|}{\partial \mathbf{x}^{\mathrm{pin}}}$ can be derived as
\begin{align}
\textstyle\frac{\partial \|\bm{\psi}_k-\bm{\psi}_{n_i,m}^{\mathrm{pin}}\|}{\partial \mathbf{x}^{\mathrm{pin}}}=\left[\frac{\partial \|\bm{\psi}_k-\bm{\psi}_{n_i,m}^{\mathrm{pin}}\|}{\partial x_{n_1,1}^{\mathrm{pin}}},\cdots,\frac{\partial \|\bm{\psi}_k-\bm{\psi}_{n_i,m}^{\mathrm{pin}}\|}{\partial x_{n_K,|\tilde{\mathcal{M}}_{n_K}|}^{\mathrm{pin}}}\right],    
\end{align}
where
\begin{align}
\textstyle\frac{\partial \|\bm{\psi}_k-\bm{\psi}_{n_i,m}^{\mathrm{pin}}\|}{\partial x_{n_k,j}^{\mathrm{pin}}}=\left\{
\begin{array}{ll}
\frac{x_{n_i,m}^{\mathrm{pin}}-x_k}{\|\bm{\psi}_k-\bm{\psi}_{n_i,m}^{\mathrm{pin}}\|} & {\mathrm{if}}\ n_i = n_k\ \mathrm{and}\ m=j\\
0 & {\mathrm{otherwise}}.
\end{array}
\right.    
\end{align}
Based on \eqref{WG_channel} and $\|\bm{\psi}_{n_i,m}^{\mathrm{pin}}-\bm{\psi}_{n_i,0}^{\mathrm{pin}}\|=x_{n_i,m}^{\mathrm{pin}}$, the term $\frac{\partial g_{n_i,m,k}}{\partial \mathbf{x}^{\mathrm{pin}}}$ in \eqref{A_grd_x} can be derived as
\begin{align}
    \textstyle\frac{\partial g_{n_i,m,k}}{\partial \mathbf{x}^{\mathrm{pin}}}=-j\frac{2\pi n_{\mathrm{eff}}g_{n_i,m,k}}{\lambda}\cdot \frac{\partial x_{n_i,m}^{\mathrm{pin}}}{\partial \mathbf{x}^{\mathrm{pin}}},
\end{align}
where
\begin{align}
\textstyle\frac{\partial x_{n_i,m}^{\mathrm{pin}}}{\partial \mathbf{x}^{\mathrm{pin}}}=    \left[\frac{\partial x_{n_i,m}^{\mathrm{pin}}}{\partial x_{n_1,1}^{\mathrm{pin}}},\cdots,\frac{\partial x_{n_i,m}^{\mathrm{pin}}}{\partial x_{n_K,|\tilde{\mathcal{M}}_{n_K}|}^{\mathrm{pin}}}\right]
\end{align}
and
\begin{align}\label{grd_end_x}
\textstyle\frac{\partial x_{n_i,m}^{\mathrm{pin}}}{\partial x_{n_k,j}^{\mathrm{pin}}}=\left\{
\begin{array}{ll}
1 & {\mathrm{if}}\ n_i = n_k\ \mathrm{and}\ m=j\\
0 & {\mathrm{otherwise}}.
\end{array}
\right.    
\end{align}
At last, $\bm{\omega}_{\mathbf{x}^{\mathrm{pin}}}(\tau)$ can be obtained by combining \eqref{grd_start_x}--\eqref{grd_end_x}.
\end{appendices}


\begin{thebibliography}{10}

\bibitem{DingFlexible}
Z. Ding, R. Schober, and H. V. Poor, ``Flexible-antenna systems: A pinching-antenna perspective,'' \textit{IEEE Trans. Commun.}, Mar. 2025. [Online]. Available: DOI: 10.1109/TCOMM.2025.3555866.

\bibitem{LiuPinching}
Y. Liu, Z. Wang, X. Mu, C. Ouyang, X. Xu, and Z. Ding, ``Pinching-antenna systems (PASS): Architecture designs, opportunities, and outlook,'' \textit{IEEE Commun. Mag.}, Sep. 2025. [Online]. Available: DOI: 10.1109/MCOM.001.2500037.

\bibitem{WongFluid}
K.-K. Wong, K.-F. Tong, Y. Zhang, and Z. Zheng, ``Fluid antenna system for 6G: When Bruce Lee inspires wireless communications,'' \textit{Elect. Lett.,} vol. 56, no. 24, pp. 1288--1290, Nov. 2020.

\bibitem{ZhuModeling}
L. Zhu, W. Ma, and R. Zhang, ``Modeling and performance analysis
for movable antenna enabled wireless communications,'' \textit{IEEE Trans. Wireless Commun.,} vol. 23, no. 6, pp. 6234--6250, 2023.

\bibitem{XuRate}
Y. Xu, Z. Ding, and G. K. Karagiannidis, ``Rate maximization for downlink pinching-antenna systems,'' \textit{IEEE Wireless Commun. Lett.,} vol. 14, no. 5, pp. 1431--1435, May 2025.

\bibitem{QinJoint}
Y. Qin, Y. Fu, and H. Zhang, ``Joint antenna position and transmit power optimization for pinching antenna-assisted ISAC systems,'' \textit{IEEE Wireless Commun. Lett.,} Aug. 2025. [Online]. Available: DOI: 10.1109/LWC.2025.3597719.

\bibitem{XieA}
X. Xie, F. Fang, Z. Ding, and X. Wang, ``A low-complexity placement design of pinching-antenna systems,'' \textit{IEEE Commun. Lett.,} vol. 29, no. 8, pp. 1784--1788, Aug. 2025.

\bibitem{DingAna}
Z. Ding and H. V. Poor, ``Analytical optimization for antenna placement in pinching-antenna system,'' Jul. 2025, \textit{arXiv:2507.13307}.

\bibitem{ChenDynamic}
J.-C. Chen, P.-C. Wu and K.-K. Wong, ``Dynamic placement of pinching antennas for multicast MU-MISO downlinks,'' \textit{IEEE Open J. Commun. Soc.,} vol. 6, pp. 5611--5625, Jun. 2025.

\bibitem{FuPower}
Y. Fu, F. He, Z. Shi, and H. Zhang, ``Power minimization for NOMA-assisted pinching antenna systems with multiple waveguides,'' Mar. 2025, \textit{arXiv:2503.20336v2}.


\bibitem{ZhouA}
K. Zhou, W. Zhou, D. Cai, X. Lei, Y. Xu, Z. Ding, and P. Fan, ``A gradient meta-learning joint optimization for beamforming and antenna position in pinching-antenna systems,'' Jun. 2025, \textit{	arXiv:2506.12583}.

\bibitem{WangModeling}
Z. Wang, C. Ouyang, X. Mu, Y. Liu, and Z. Ding, ``Modeling and
beamforming optimization for pinching-antenna systems,'' Feb. 2025, \textit{arXiv:2502.05917}.

\bibitem{PozarMicro}
D. M. Pozar, \textit{Microwave Engineering}, 4th ed. Wiley, New York, US, 1998.

\bibitem{ZhangBeam}
H. Zhang, N. Shlezinger, F. Guidi, D. Dardari, M. F. Imani, and Y. C. Eldar, ``Beam focusing for near-field multiuser MIMO communications,'' \textit{IEEE Trans. Wireless Commun.,} vol. 21, no. 9, pp. 7476--7490, Sept. 2022.

\bibitem{ShenFrac}
K. Shen and W. Yu, ``Fractional programming for communication systems—Part II: Uplink scheduling via matching,'' \textit{IEEE Trans. Signal Process.,} vol. 66, no. 10, pp. 2631--2644, May 2018.

\bibitem{ShenFrac1}
K. Shen and W. Yu, ``Fractional programming for communication systems—Part I: Power control and beamforming,'' \textit{IEEE Trans. Signal Process.,} vol. 66, no. 10, pp. 2616--2630, May 2018.






\end{thebibliography}
\end{document}